\documentclass[a4paper,twocolumn,notitlepage,aps,pra,10pt]{revtex4-2}
\usepackage{bbm}
\usepackage{amsmath}
\usepackage{amssymb}
\usepackage{amsthm}
\usepackage{amsfonts}
\usepackage{amsmath}
\usepackage{graphicx}
\usepackage{xcolor}
\usepackage{csquotes}
\usepackage{braket}
\usepackage{setspace}
\usepackage{verbatim}
\usepackage[colorlinks=true,linkcolor=teal,citecolor=teal,urlcolor=teal]{hyperref}
\usepackage{mathtools}
\usepackage{booktabs}
\usepackage{braket}
\usepackage{mleftright}
\usepackage{siunitx}
\DeclareSIUnit\hartree{Ha}

\newtheorem{proposition}{Proposition}

\theoremstyle{definition}

\begin{document}

\title{Matrix product channel: Variationally optimized quantum tensor network\\ to mitigate noise and reduce errors for the variational quantum eigensolver}

\author{Sergey Filippov}
\email{sergey.filippov@algorithmiq.fi}
\author{Boris Sokolov}
\author{Matteo A. C. Rossi}
\author{Joonas Malmi}
\author{Elsi-Mari Borrelli}
\author{Daniel Cavalcanti}
\author{Sabrina Maniscalco}
\author{Guillermo Garc\'{\i}a-P\'{e}rez}
\affiliation{Algorithmiq Ltd, Kanavakatu 3 C, FI-00160 Helsinki, Finland}

\begin{abstract}
Quantum processing units boost entanglement at the level of hardware and enable physical simulations of highly correlated electron states in molecules and intermolecular chemical bonds. The variational quantum eigensolver provides a hardware-efficient toolbox for ground state simulation; however, with limitations in precision. Even in the absence of noise, the algorithm may result into a biased energy estimation, particularly with some shallower ansatz types. Noise additionally degrades entanglement and hinders the ground state energy estimation (especially if the noise is not fully characterized). Here we develop a method to exploit the quantum-classical interface provided by informationally complete measurements to use classical software on top of the hardware entanglement booster for ansatz- and noise-related error reduction. We use the tensor network representation of a quantum channel that drives the noisy state toward the ground one. The tensor network is a completely positive map by construction, but we elaborate on making the trace preservation condition local so as to activate the sweeping variational optimization. This method brings into reach energies below the noiseless ansatz by creating additional correlations among the qubits and denoising them. Analyzing the example of the stretched water molecule with a tangible entanglement, we argue that a hybrid strategy of using the quantum hardware together with the classical software outperforms a purely classical strategy provided the classical parts have the same bond dimension. As a byproduct we discuss the expressivity of matrix product channels and address the overfitting problem emerging in postprocessing actual measurement data. The proposed optimization algorithm extends the variety of noise mitigation methods and facilitates the more accurate study of the energy landscape for deformed molecules. The algorithm can be applied as the final postprocessing step in the quantum hardware simulation of protein–ligand complexes in the context of drug design.
\end{abstract}

\maketitle

\section{Introduction} \label{section-introduction}
The variational quantum eigensolver (VQE)~\cite{peruzzo-2014,mcclean-2016} has the potential to demonstrate the first useful advantage of near-term quantum devices and advance quantum chemistry~\cite{elfving-2020,mcardle-2020,cerezo-2021,bharti-2022,fedorov-2022,anand-2022,tilly-2022}. In particular, VQE is expected to be useful in calculation of large scale protein–ligand interaction energies~\cite{malone-2022,kirsopp-2022} and protein folding~\cite{robert-2021}. A vista is to describe new biomolecules~\cite{marchetti-2022} and rethink small-molecule drug discovery~\cite{schneider-2020} from a quantum perspective. Existing methods of classical computation struggle in describing highly entangled multipartite quantum states in quantum chemistry and even deep neural network approaches can deal with at most about 30 electrons \cite{hermann-2020}. The VQE uses an actual quantum processor to physically prepare a quantum state approximating the ground state of a given Hamiltonian, thus playing a role of an entanglement booster. However, the noise levels in existing quantum processors create one significant roadblock for using VQE to gain useful quantum advantage over classical algorithms.

Unfortunately, the fault-tolerant error-correcting quantum computers are currently unavailable and would require lots of ancillary physical qubits for exploring relevant quantum chemistry problems~\cite{goings-2022}. Other hardware-based noise mitigation methods require (ideal) quantum operations on multiple copies of the VQE state~\cite{koczor-2021,huggins-2021} and may be hard to implement for many qubits. Focusing on the very near-term quantum computing, we therefore have no other option but to consider software-based noise mitigation strategies and address the challenging problem of how to combine the noisy quantum hardware and the classical postprocessing in the most efficient way. Here we present one of several approaches incorporated in the Algorithmiq's proprietary quantum chemistry platform \emph{Aurora} to mitigate noise and reduce errors.

Noise degrades coherences in the VQE output and generally corrupts its entanglement, and the purpose of the classical postprocessing is to mitigate those effects. Many existing proposals are listed in Ref.~\cite{tilly-2022}, with some of them assuming a specific noise model (the probalisitic error cancellation for local noise~\cite{temme-2017,li-2017}, the non-local noise inversion for shallow circuits~\cite{guo-2022}, the depolarizing noise inversion~\cite{vovrosh-2021}) and some of them eliminating an unknown qubit-local Markovian noise by extrapolating the results for several noise strengths to get the noise-free estimation~\cite{endo-2018}. In the most optimistic scenario, in which the noise is completely mitigated, there is still a problem of the VQE ansatz not being precise, i.e., a deviation from the ground state energy is inevitable. This is particularly important for simulating molecules with deformed and stretched bonds as they are known to typically have a high degree of entanglement~\cite{boguslawski-2012,boguslawski-2013,molina-espiritu-2015} and, therefore, require more complicated VQE ans\"{a}tze (needless to mention, there are limitations of the classical methods to deal with such molecules). However, studying the energy landscape of deformed molecules is a cornerstone for determining the rates of chemical reactions and analyzing the binding affinity of protein–ligand complexes. The progress in this direction can dramatically reduce the cost of a \emph{de novo} drug design~\cite{malone-2022,seo-2021}.

There is a clear physical intuition about how to reduce the built-in VQE ansatz error and drive a noiseless $N$-qubit VQE ansatz state $\ket{\psi_{\rm VQE}}$ closer to the actual $N$-qubit ground state $\ket{\psi_{0}}$: this can be done by a proper $2^N \times 2^N$ unitary operator $U$ provided this operator can be efficiently optimized classically --- an assumption fulfilled in the tensor network representation~\cite{cirac-2017,luchnikov-2021,haghshenas-2022,rudolph-2022}. With the use of informationally complete measurements, the idea of optimization with classical software after creating entanglement on a quantum hardware is realizable also in the noisy scenario with a replacement of the unitary tensor network by a tensor network of completely positive and trace preserving maps \cite{vilma}. This generalization enables one both to reduce the VQE ansatz errors and to mitigate low-intensity noise by driving the mixed density operator at the VQE output closer to the pure ground state. This method makes it possible to go below the noise-free VQE ansatz energy and for this reason deserves further attention. The previously proposed tensor network for the global completely positive map has a ladder topology \cite{vilma}, which partially restricts its expressiveness. In this paper we propose to use a completely positive tensor network map of another topology, where the Kraus operators are effectively represented by the matrix product operators. This topology is known in the literature as the locally-purified density operator~\cite{werner-2016,torlai-2020} and, as we show, it yields a more expressive set of maps, provided the bond dimension is fixed. 

The essential roadblock for the variational optimization of the tensor network with topology of the locally-purified density operator is the trace-preservation condition that involves all constituent tensors and makes it nearly impossible to modify a single-qubit tensor without violating the trace preservation condition for the whole map~\cite{srinivasan-2021}. In Refs.~\cite{torlai-2020} and \cite{guo-2022}, the \emph{ad hoc} solution was to introduce the additional penalty in the cost function so as to penalize a deviation from the trace preservation. However, this approach still results in an error even after convergence~\cite{torlai-2020}. As insignificant as this error can be for the purposes of the quantum channel tomography it is generally too high and unpredictable for this approach to be used in solving the quantum chemistry problems, where high accuracy is needed. 

Here, we develop an alternative strategy to the trace preservation resulting in a linear equation for a single-qubit tensor. The linearity enables us to incorporate this condition in the semidefinite programming problem~\cite{BVbook,SCbook} for a given single-qubit tensor, upon solving which we can switch to another qubit and so on --- a `sweeping' optimization envisioned in Ref.~\cite{srinivasan-2021}. Therefore, we eliminate the trace-preservation roadblock for this tensor network with topology of the locally-purified density operator (that defines a completely positive map by construction), and for the sake of brevity we refer to this map as a \emph{matrix product channel} (MPC). 

By variationally optimizing MPC, we demonstrate a proof-of-principle simulation of how the classical postprocessing on top of the hardware entanglement booster can facilitate mitigating noise and reducing errors in the VQE. We show that the ``quantum+classical'' strategy is advantageous as compared to the purely ``classical'' strategy based on the conventional density matrix renormalization group (DMRG) \cite{schollwock-2011} estimation of the ground energy with the same bond dimension as in the MPC.

The paper is organized as follows. In Sec.~\ref{section-noise-and-errors}, we briefly review the ansatz-related and noise-related errors in the VQE as well as desribe the classical postprocessing strategy in general. In Sec.~\ref{section-mpc}, the MPC tensor network structure is outlined. In Sec.~\ref{section-expressibility}, we address the issue of how expressive the MPC tensor network is. In Sec.~\ref{section-tp} we elaborate the trace preservation condition for the MPC. In Sec.~\ref{section-variational-algorithm}, the MPC variational optimization algorithm is given. In Sec.~\ref{section-results}, we apply the developed algorithm and amend the VQE estimation of the ground energy for the illustrative example of the stretched water molecule. In Sec.~\ref{section-conclusions}, brief conclusions are given. 

\section{Noise and errors in the variational quantum eigensolver} \label{section-noise-and-errors}

\begin{figure}[b]
    \centering
    \includegraphics[width = 8.5cm]{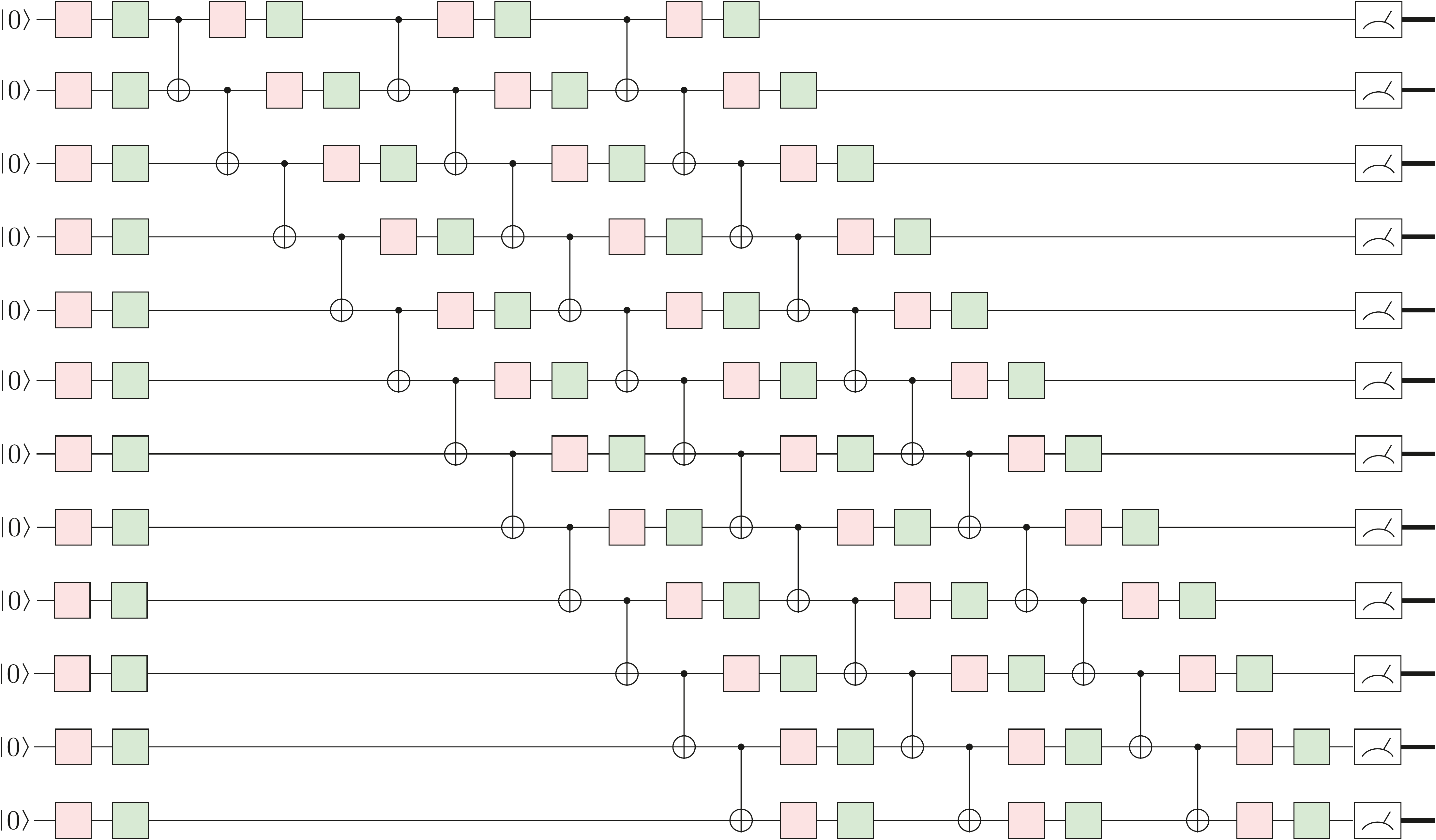}
    \caption{Hardware-efficient VQE ansatz with 12 qubits simulating the ground state of the stretched ${\rm H}_2{\rm O}$ molecule. Pink and light-green squares correspond to single-qubit conventional rotations around $y$ and $z$ axes, respectively.}
    \label{figure-circuit}
\end{figure}

The VQE's applicability to quantum chemistry tasks is most straightforward as it aims at preparing a close analogue $\ket{\psi_{\rm VQE}}$ to the ground state $\ket{\psi_{0}}$ for a given effective Hamiltonian $H$ describing a selected number of spin-orbitals~\cite{kandala-2017,nam-2020}. In this study, we do not consider the approximation error in defining the effective Hamiltonian and assume that this error is negligible as compared to the error of the VQE ansatz, 
\begin{equation}
\epsilon_{\rm VQE} = \bra{\psi_{\rm VQE}} H \ket{\psi_{\rm VQE}} - \bra{\psi_{0}} H \ket{\psi_{0}} > 0.
\end{equation}

\noindent Fig.~\ref{figure-circuit} illustrates a simple quantum circuit that provides a hardware-efficient~\cite{kandala-2017} $12$-qubit state $\ket{\psi_{\rm VQE}}$ approximating the ground state $\ket{\psi_0}$ of the stretched ${\rm H}_2{\rm O}$ molecule modelled by a 6-molecular-orbital Hamiltonian $H$.

If current quantum computers were noiseless, $\epsilon_{\rm VQE}$ would be the only estimation error for the ground state energy; however, since all the quantum gates in the circuit as well as state preparations and measurements are generally imperfect, we actually deal with a mixed state $\varrho \neq \ket{\psi_{\rm VQE}}\bra{\psi_{\rm VQE}}$ so that the deviation 
\begin{equation}
\epsilon_{\rm noise} = {\rm tr}[\varrho H] - \bra{\psi_{\rm VQE}} H \ket{\psi_{\rm VQE}}
\end{equation}

\noindent may significantly contribute to the total error (especially in the case of deep quantum circuits). There are several approaches addressing the latter type of noise-induced error~\cite{tilly-2022,endo-2021}; however, only a few approaches are capable of amending the first one, e.g., the Lanczos algorithm~\cite{suchsland-2021} and the symmetry constraints on the Hamiltonian $H$, which \emph{per se} enable one to implement a low-cost error mitigation~\cite{endo-2021,bonet-monroig-2018}. Omitting the symmetry considerations for the sake of universality, we are left with the important problem of how to simultaneously diminish both types of errors.

In an actual experiment one has no access to the states $\ket{\psi_{\rm VQE}}$ or $\ket{\psi_0}$ but deals with the only information available, namely, the measurement outcomes (see Figs.~\ref{figure-circuit} and \ref{figure-setup}). If the measurements themselves are informationally complete, then a modest number of shots (that does not scale exponentially in the number of qubits) is still enough to estimate the energy~\cite{garcia-perez-2021}, with the estimate being unbiased, i.e., on average we get the error $\epsilon_{\rm VQE} + \epsilon_{\rm noise}$. The idea of this paper is to present a further development of the classical algorithm~\cite{vilma} that processes the measurement outcomes in such a way that the error in energy estimation can be reduced from $\epsilon_{\rm VQE} + \epsilon_{\rm noise}$ to below $\epsilon_{\rm VQE}$, i.e., not only the noise is mitigated but also the ansatz itself is corrected. 

\begin{figure}
    \centering
    \includegraphics[width = 8cm]{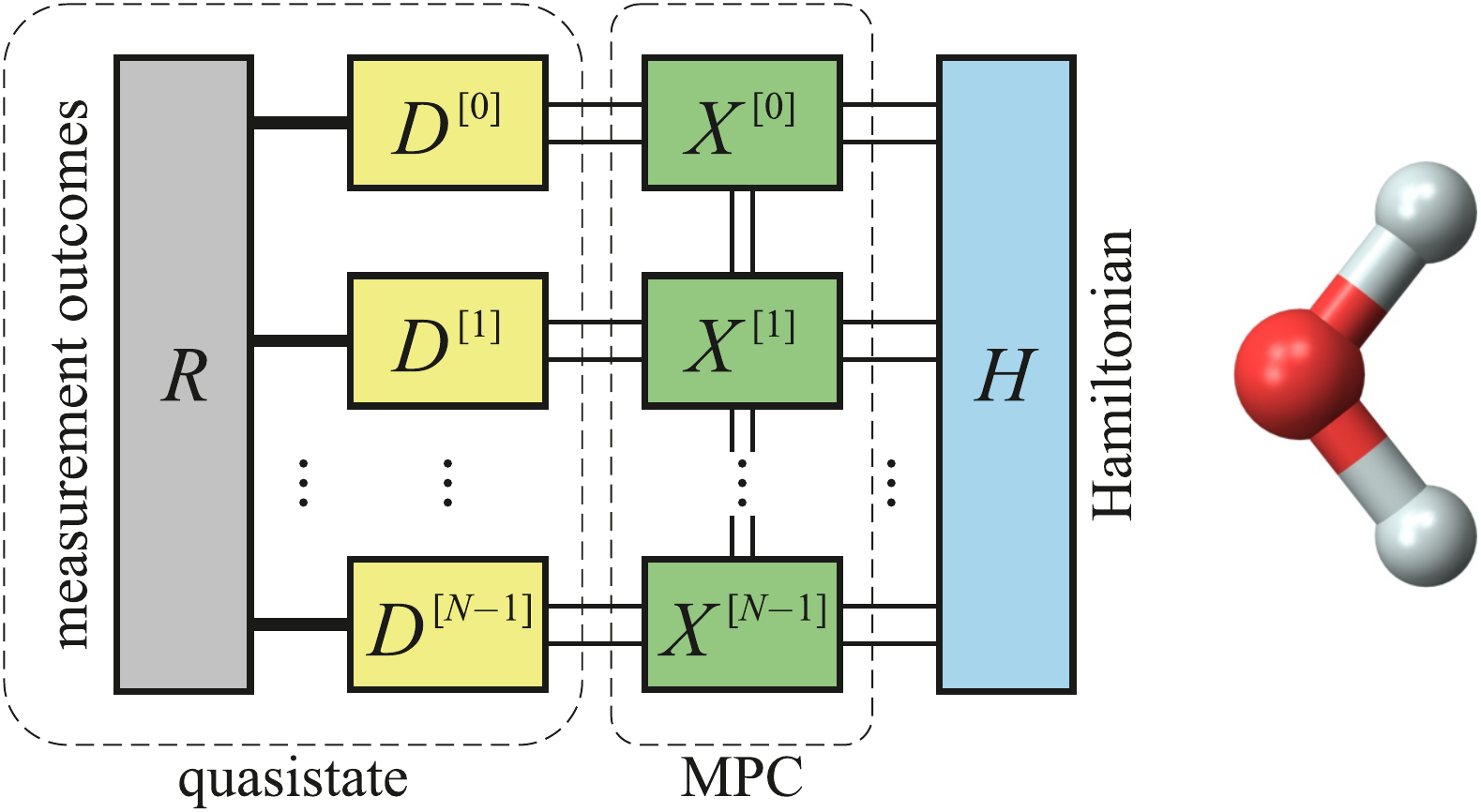}
    \caption{Classical postprocessing of measurement outcomes with a matrix product channel (MPC) to get a better estimation of the molecule ground-state energy.}
    \label{figure-setup}
\end{figure}

The key components of the postprocessing task are presented in Fig.~\ref{figure-setup}. $R$ is a tensor with a finite collection of $S$ measurement outcomes; $D^{[m]}$ is a tensor encoding a set of dual operators for the effects in the informationally complete measurement of $m$th qubit (enumerated from $0$ to $N-1$). The contraction of all tensors $R$ and $\{D^{[m]}\}_{m=0}^{N-1}$ yields an estimate $\varrho$ of the state (which we refer to as the quasistate) that becomes an actual density operator (describing the VQE ansatz affected by noise) in the limit of infinitely many shots. MPC is a tensor network describing a completely positive and trace preserving map $\Phi$ to be variationally optimized so as to minimize the energy ${\rm tr}\{\Phi[\varrho]H\}$, where $H$ is the system effective Hamiltonian. 

Importantly, the postprocessing algorithm that we develop does not assume the noise to be known and, therefore, is beyond recently discussed limitations of quantum error mitigation~\cite{quek-2022}. By analyzing the example of the stretched ${\rm H}_2{\rm O}$ molecule we also justify that the advantage in mitigating noise and reducing the overall energy estimation error down to $\epsilon_{\rm q+cl}$ is indeed obtained by a collaborative effort of the quantum entanglement booster and the classical postprocessing in the form of an optimized MPC tensor network (with a fixed classical bond dimension $r=2$). The classical DMRG simulation of the ground state with the same bond dimension $r=2$ (performed alone without the entanglement booster in the form of the VQE) is not able to reach the same accuracy in the energy estimation and leads to an error $\epsilon_{\rm cl} > \epsilon_{\rm q+cl}$.

\section{Matrix product channel} \label{section-mpc}

Operators and maps that are widely used in quantum information science take an impressively clear geometrical form in the tensor network formalism. Consider a $d$-dimensional quantum system and an orthonormal basis $\{\ket{i}\}_{i=0}^{d-1}$ in the associated Hilbert space. A density operator $\varrho$ is defined by its matrix elements $\bra{i} \varrho \ket{i'}$ that altogether form a rank-2 tensor $\varrho_{ii'}$, which is nothing else but the positive semidefinite matrix with unit trace. To distinguish ket-indices $\{i\}_{i=0}^{d-1}$ from bra-indices $\{i'\}_{i'=0}^{d-1}$ in the tensor diagram for $\varrho = \sum_{i,i'} \varrho_{ii'} \ket{i} \bra{i'}$ we mark ket- and bra-legs of a tensor by outgoing and incoming arrows, respectively. The connected legs are summed over, so the introduced arrows also indicate the matrix multiplication order. The trivial diagrams for $\varrho$ and ${\rm tr}[A \varrho B]$ are shown in Figs.~\ref{figure-diagrams}(a) and \ref{figure-diagrams}(b). 

\begin{figure*}
    \centering
    \includegraphics[width = 14cm]{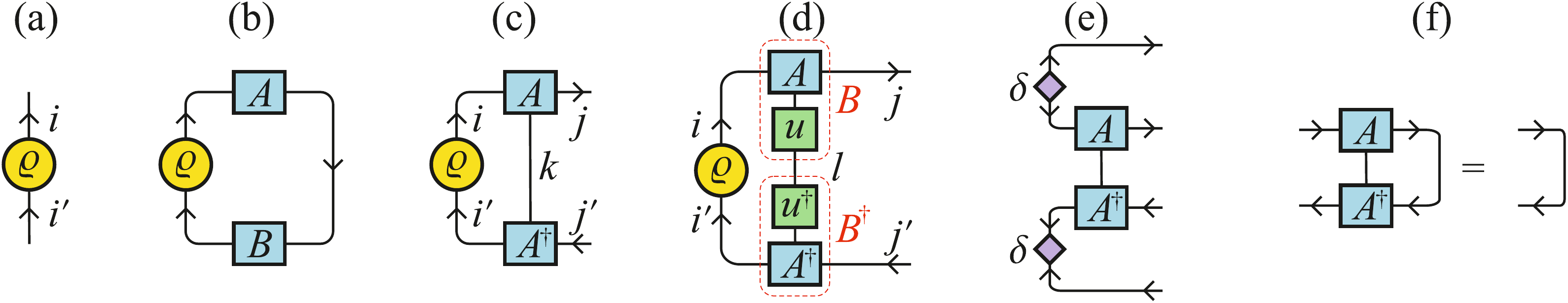}
    \caption{Basic tensor diagrams involving a density operator and a quantum channel. (a) Density operator. (b) Rank-0 tensor ${\rm tr}[A \varrho B]$. (c) Quantum channel output $\sum_{k} A_k \varrho A_{k}^{\dag}$. (d) Nonuniqueness of the Kraus operators. (e) Choi operator. (f) Trace preservation condition.}
    \label{figure-diagrams}
\end{figure*}

\begin{figure*}
    \centering
    \includegraphics[width = 18cm]{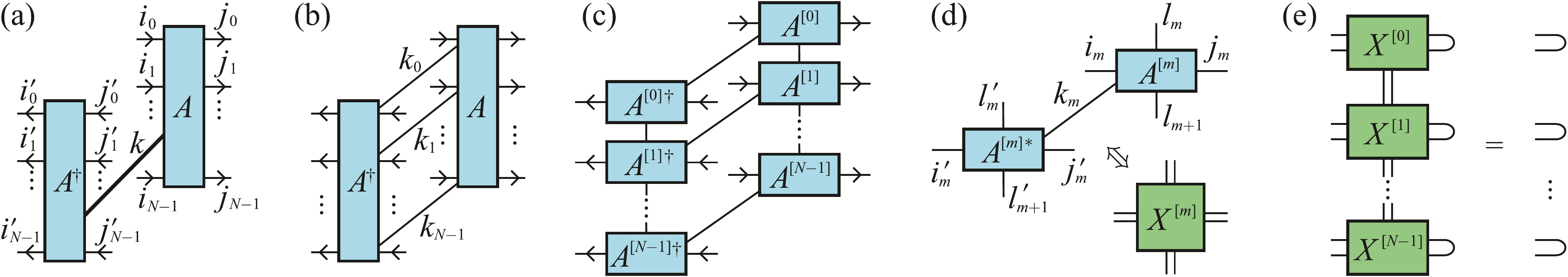}
    \caption{Matrix product channel deduction. (a) Multipartite Kraus operators. (b) Splitting the Kraus bond index into a multiindex. (c) Matrix product representation of the Kraus operators. (d) Core of the matrix product channel. (e) Trace preservation condition.}
    \label{figure-mpc}
\end{figure*}

A completely positive map $\Phi$ adopts the diagonal-sum representation $\Phi [\varrho] = \sum_{k=0}^{\widetilde{K}-1} A_k \varrho A_k^{\dag}$~\cite{holevo-2019}. A collection of $\widetilde{K}$ Kraus operators, $\{A_k\}_{k=0}^{\widetilde{K}-1}$, can be viewed as a rank-3 tensor $A_{ijk}$ with two indices $i$ and $j$ of physical dimension $d$ and one index $k$ of the so-called Kraus dimension $\widetilde{K}$~\cite{torlai-2020,guo-2022}. A tensor diagram for $\Phi [\varrho]$ is depicted in Fig.~\ref{figure-diagrams}(c). Some basic properties such as non-uniqueness of the Kraus operators immediately follow from this tensor network representation. Indeed, let $u$ be an isometric matrix such that $u^{\dag} u = I$. Incorporating the very identity operator into the connected line for Kraus indices as shown in Fig.~\ref{figure-diagrams}(d) and contracting tensors $A$ and $u$, we get new Kraus operators $B_l = \sum_k u_{lk} A_k$ for the same map $\Phi$, i.e., $\Phi [\varrho] = \sum_{l} B_l \varrho B_l^{\dag}$. Figure \ref{figure-diagrams}(e) shows a tensor diagram for the Choi operator $\Phi \otimes {\rm Id} [\ket{\delta}\bra{\delta}]$, where ${\rm Id}$ stands for the identity transformation and $\ket{\delta} = \sum_{i = 0}^{d} \ket{ii}$ is an unnormalized maximally entangled state. Considering an eigendecomposition of the Choi operator, one can readily see that the number of Kraus operators can always be decreased to be less than or equal to $d^2$. If the considered completely positive map $\Phi$ additionally enjoys the trace preservation property $\sum_{k=0}^{\widetilde{K}-1} A_k^{\dag} A_k = I$, then $\Phi$ is a quantum channel. In terms of the dual map $\Phi^{\dag}$ (defined through the mathematical identity ${\rm tr}\big[ \Phi^{\dag}[X]Y \big] = {\rm tr} \big[X \Phi[Y] \big]$ for all $X$, $Y$) the trace preservation condition reads $\Phi^{\dag}[I] = I$ meaning the dual map is unital. A tensor diagram for the trace preservation condition is shown in Fig.~\ref{figure-diagrams}(f).

Having at hand the tensor network representation of a quantum channel for a single $d$-dimensional quantum system, we can now readily generalize it to a multipartite scenario with $N$ quantum systems. We use the subscript $m = 0, \ldots, N-1$ in physical indices $i_m$ and $j_m$ to address individual quantum systems. In our diagrams, the constituent systems are stacked vertically (see Fig.~\ref{figure-mpc}) so the index $m$ enumerates \textit{floors} associated with the systems. A general Kraus operator $A_{i_0 \ldots i_{N-1} j_0 \ldots j_{N-1} k}$ is a tensor of rank $2N+1$ as it is shown in Fig.~\ref{figure-mpc}(a); however, the total number of Kraus operators can be as big as $d^{2N}$. Reshaping the Kraus index $k$ into a multiindex $(k_0 \ldots k_{N-1})$, where the dimension of each subindex $k_m$ is $d^2$, we get a $3N$-rank tensor $A_{i_0 \ldots i_{N-1} j_0 \ldots j_{N-1} k_0 \ldots k_{N-1}}$, see Fig.~\ref{figure-mpc}(b). Exploiting the full analogy with the matrix product states and the matrix product operators~\cite{schollwock-2011,orus-2014,cirac-2021}, we rewrite the obtained Kraus operator in the following matrix product form:
\begin{equation} \label{matrix-product-Kraus-operator}
    A_{i_0 \ldots i_{N-1} j_0 \ldots j_{N-1} k_0 \ldots k_{N-1}} = A_{i_0 j_0 k_0}^{[0]} \cdots A_{i_{N-1} j_{N-1} k_{N-1}}^{[{N-1}]}.
\end{equation}
Fig.~\ref{figure-mpc}(c) illustrates the resulting tensor diagram for a multiparte quantum channel which we refer to as a \emph{matrix product channel} (MPC). The channel output $\Phi[\varrho]$ and the Choi operator for the channel are known in the literature as a locally-purified density operator~\cite{werner-2016,torlai-2020}.

We denote by $l_m$ the virtual indices in between the tensors $A_{i_{m-1} j_{m-1} k_{m-1}}^{[m-1]}$ and $A_{i_m j_m k_m}^{[m]}$, see Fig.~\ref{figure-mpc}(d). Although the matrix product decomposition~\eqref{matrix-product-Kraus-operator} always exists, in the worst-case scenario the bond dimension $\vert \{l_m\} \vert$ grows exponentially with the increase of $m$ as $\min(d^{4m},d^{4(N-m)})$. Artificially restricting the bond dimension to take some fixed value $r \ll \min(d^{4m},d^{4(N-m)})$, we reduce the class of simulable completely positive maps on the one hand but make the tensor network computationally efficient on the other hand. The efficacy originates from the fact that we use rank-5 tensors $A_{i_m j_m l_m l_{m+1} k_m}^{[m]}$ and their complex conjugated versions $A_{i'_m j'_m l'_m l'_{m+1} k_m}^{[m]\ast}$ that are contracted along the indices of small dimensions (independent of $N$). Therefore, the complexity of such a description grows linearly with the number of systems $N$. Moreover, the map constructed in such a way is automatically completely positive as it adopts the diagonal-sum representation. 

The efficacy of the described tensor network was reported previously for simulating open quantum many-body systems with near-neighbour Hamiltonian and dissipator terms~\cite{werner-2016} as well as in quantum tomography~\cite{torlai-2020}. The matter is that a relatively small bond dimension $r$ is enough in many physically relevant situations thanks to the local nature of interactions. In this paper, we provide an affirmative answer to the question whether the constructed multipartite map with a small bond dimension $r$ is able to mitigate noise and reduce errors for the variational quantum eigensolver. 

\begin{figure*}
    \centering
    \includegraphics[width = 15cm]{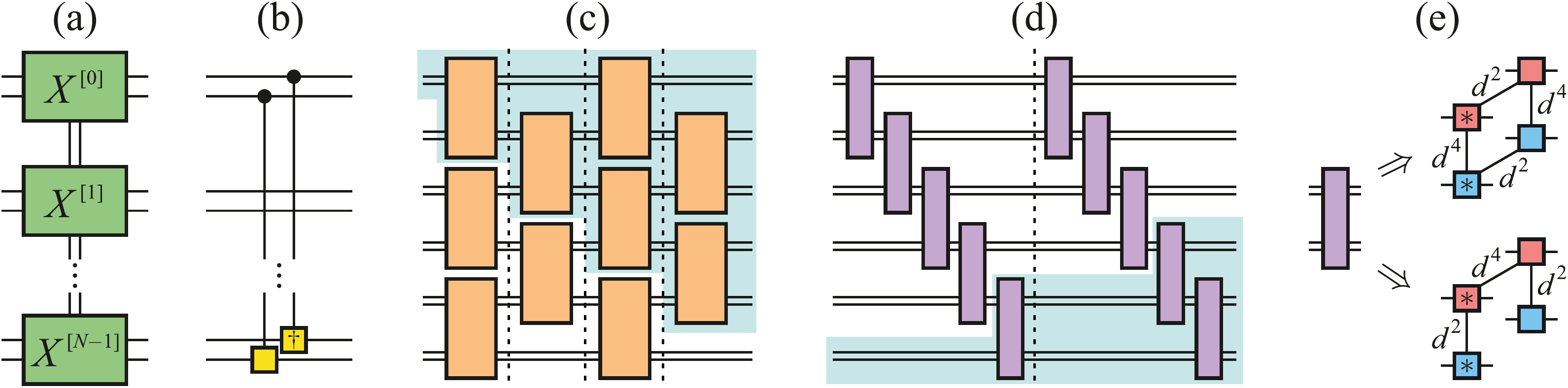}
    \caption{Matrix product channel vs local tensor networks. (a) Matrix product channel. (b) Example of the controlled unitary gate. (c) Brickwall tensor network of quantum channels. Dotted vertical lines indicate layers. Causal cones are marked by shaded areas. (d) Ladder tensor network of quantum channels. (e) Decompositions of a bipartite quantum channel. Dimensions of connected lines are explicitly given.}
    \label{figure-other-maps}
\end{figure*}

The appealing beauty of the built-in complete positivity is overshadowed by the trace preservation condition $\Phi^{\dag}[I^{\otimes N}] = I^{\otimes N}$ which becomes notoriously difficult to deal with~\cite{srinivasan-2021}. (This resembles an interplay of complexities for the complete positivity condition and the trace preservation condition in memoryless and memory-kernel master equations: whenever one of the conditions is easy to verify, the other one turns out to be a challenge~\cite{chruscinski-2022}.) The matter is that the trace preservation condition involves all the subsystems as a whole [see Fig.~\ref{figure-mpc}(e)] so that even if we modify $m$th party tensors $A_{i_m j_m l_m l_{m+1} k_m}^{[m]}$ and $A_{i'_m j'_m l'_m l'_{m+1} k_m}^{[m]\ast}$ only (for a fixed $m$), then we not only disturb the identity operator at the $m$-th floor of the diagram but also potentially disturb identity operators for all other $N-1$ floors. In fact, a roadblock for the variational optimization is how the tensors at one floor of the network may be varied without violating the trace preservation constraint~\cite{srinivasan-2021}. In Refs.~\cite{torlai-2020} and~\cite{guo-2022}, the ad-hoc solution was to introduce the additional penalty in the cost function; however, this approach may still lead to the error $\frac{1}{\sqrt{d^N}} \| \Phi^{\dag}[I] - I \|_2 \approx 0.1$ after convergence~\cite{torlai-2020} which is insufficient for quantum chemistry problems. In Sec.~\ref{section-tp} we move away that roadblock by reformulating the trace preservation constraint in the form of a linear condition on the tensor $X_{i_m j_m l_m l_{m+1} i'_m j'_m l'_m l'_{m+1}}^{[m]} = \sum_{k_m} A_{i_m j_m l_m l_{m+1} k_m}^{[m]} A_{i'_m j'_m l'_m l'_{m+1} k_m}^{[m]\ast}$ and reducing the optimization problem to a semidefinite programming problem for the positive semidefinite matrix $X_{(i_m j_m l_m l_{m+1}), (i'_m j'_m l'_m l'_{m+1})}^{[m]}$.

\section{Expressivity of matrix product channels} \label{section-expressibility}

By ${\cal E}_{{\rm MPC}}(r,K)$ denote the set of channels attainable within the matrix product channel ansatz with the bond dimension $r$ and the Kraus dimension $K$ (per single subsystem). Clearly, the inclusion relation ${\cal E}_{{\rm MPC}}(r,K) \subseteq {\cal E}_{{\rm MPC}}(r',K')$ takes place if $r \leq r'$ and $K \leq K'$. It is worth noticing that the bond dimension $r$ and the Kraus dimension $K$ in a matrix product channel are not totally independent. A too large Kraus dimension can always be reduced to a smaller one if the bond dimension is small. In fact, the Kraus dimension at an $m$th floor can be made as small as the rank of the matrix $X_{(i_m j_m l_m l_{m+1}), (i'_m j'_m l'_m l'_{m+1})}^{[m]}$, which in turn is bounded from above by $d^2 r^2$. The reshaped eigenvectors of the matrix $X_{(i_m j_m l_m l_{m+1}), (i'_m j'_m l'_m l'_{m+1})}^{[m]}$ are the renewed tensors $\{\widetilde{A}_{i_m j_m l_m l_{m+1} \widetilde{k_m}}^{[m]}\}_{\widetilde{k_m} = 0}^{{\rm rank} X^{[m]}  - 1}$ enabling a more efficient representation. Therefore, we have just proved the following result.
\begin{proposition} \label{prop-decrease-Kraus}
${\cal E}_{{\rm MPC}}(r,K) \subseteq {\cal E}_{{\rm MPC}}(r,d^{2}r^{2})$ for all $K$. 
\end{proposition}
Proposition~\ref{prop-decrease-Kraus} shows that the Kraus dimension $K$ is less important in the analysis of expressivity than the bond dimension $r$. The set ${\cal E}_{{\rm MPC}}(d^{2N},d^2)$ coincides with the set of all $N$-partite channels. Imposing the restriction $r < d^{2N}$ we narrow the set of simulable maps. The natural question of how the set ${\cal E}(r,K)$ is related with other tensor network constructions proposed earlier. 

Consider a brickwall tensor network of quantum channels acting on two adjacent subsystems, see Fig.~\ref{figure-other-maps}(c). By ${\cal E}_{{\rm brick}}(L)$ denote a set of simulable $N$-partite channels in such a tensor network with $L$ layers. Then the following relation holds.
\begin{proposition} \label{prop-expressiveness-brickwall}
${\cal E}_{{\rm brick}}(L) \subseteq {\cal E}_{{\rm MPC}}(d^{L},d^{2L})$ if $L$ is even. ${\cal E}_{{\rm brick}}(L) \subseteq {\cal E}_{{\rm MPC}}(d^{L+1},d^{2(L+1)})$ and ${\cal E}_{{\rm brick}}(L) \subseteq {\cal E}_{{\rm MPC}}(d^{L+3},d^{2L})$ if $L$ is odd. 
\end{proposition}
\begin{proof}
Each two-partite channel can be decomposed down to the tensor networks shown in the bottom part of Fig.~\ref{figure-other-maps}(e). Contracting the elementary tensors within floors, we get a matrix product channel with the stated parameters $r$ and $K$. The second estimation in the case of odd $L$ results from decomposing the last layer as it is shown in the top part of Fig.~\ref{figure-other-maps}(e). 
\end{proof}

Consider a ladder tensor network of quantum channels acting on two adjacent subsystems, see Fig.~\ref{figure-other-maps}(d). By ${\cal E}_{{\rm ladder}}(L)$ denote a set of simulable $N$-partite channels in such a tensor network with $L$ ladders. Then the following relation holds.
\begin{proposition} \label{prop-expressiveness-ladder}
${\cal E}_{{\rm ladder}}(L) \subseteq {\cal E}_{{\rm MPC}}(d^{2L},d^{4L})$. 
\end{proposition}
\begin{proof}
From the viewpoint of deriving the matrix product channel through decompositions and contractions of individual channels, a single ladder layer is equivalent to two brickwall layers.
\end{proof}

To fairly compare the expressivity of a matrix product channel with that of a brickwall tensor network and a ladder tensor network we should stick to the same resulting bond dimension $r$ for all these tensor networks. Suppose $r = d^{2L}$ and $L \ll N$. In this case neither of the sets ${\cal E}_{{\rm MPC}}(d^{2L},d^{4L+2})$, ${\cal E}_{{\rm brick}}(2L)$, ${\cal E}_{{\rm ladder}}(L)$ covers the set of all $N$-partite channels. From Propositions~\ref{prop-expressiveness-brickwall} and \ref{prop-expressiveness-ladder} it follows that ${\cal E}_{{\rm MPC}}(d^{2L},d^{4L+2}) \supseteq {\cal E}_{{\rm brick}}(2L)$ and ${\cal E}_{{\rm MPC}}(d^{2L},d^{4L+2}) \supseteq {\cal E}_{{\rm ladder}}(L)$; however, we still need to explore if the sets coincide or not. To shed light on this question consider a controlled-unitary map ${\rm CU}_{m_1 m_2}$ between to distant subsystems $m_1$ and $m_2$ which leaves all other subsystems unaffected [see Fig.~\ref{figure-other-maps}(b) for the case $m_1 = 0$ and $m_2 = N-1$ and keep in mind a controlled-unitary counterpart with exchanged controlling and controlled subsystems]. Clearly, ${\rm CU}_{m_1 m_2} \in {\cal E}_{{\rm MPC}}(d,1) \subseteq {\cal E}_{{\rm MPC}}(d^{2L},d^{4L+2})$ [cf. Figs.~\ref{figure-other-maps}(a) and \ref{figure-other-maps}(b)]. However, ${\rm CU}_{m_1 m_2} \not\in {\cal E}_{{\rm brick}}(2L)$ if $L < |m_1 - m_2|$ because the causal cone of the $m_1$-th system does not contain the $m_2$-th system and vice versa. Similarly, ${\rm CU}_{m_1 m_2} \not\in {\cal E}_{{\rm ladder}}(L)$ if $L < |m_1 - m_2|$. Therefore, ${\cal E}_{{\rm MPC}}(d^{2L},d^{4L+2}) \supseteq {\cal E}_{{\rm brick}}(2L)$ and ${\cal E}_{{\rm MPC}}(d^{2L},d^{4L+2}) \supseteq {\cal E}_{{\rm ladder}}(L)$ which means that the matrix product channels are more expressive in generating different maps than the brickwall or ladder tensor networks. Continuing the same line of arguing, we come to the conclusion that the matrix product channels are more expressive than any tensor network composed of local maps (acting nontrivially on several adjacent systems) provided the bond dimensions are the same. The conclusion remains valid even if one replaces networks with open boundary conditions by those with periodic boundary conditions.

Suppose an MPC with the bond dimension $r_{\rm MPC}$ is applied to a pure matrix product state with the bond dimension $r_{\rm MPS}$ or a mixed state in the form of the density matrix product operator~\cite{verstraete-2004,zwolak-2004} with the bond dimension $r_{\rm MPDO}$ per ket and bra side of the network. Then the application of the MPC also results in the MPDO tensor network~\cite{werner-2016} with the bond dimension being $r_{\rm MPC} r_{\rm MPDO}$ ($r_{\rm MPC} r_{\rm MPS}$ if the input state is pure). The advantage of the MPC originates in the ability to add correlations and purify them. 

If the MPC is applied to a factorized input state ($r_{\rm MPDO} = 1$), then we actually deal with a purely classical optimization, where all correlations among qubits are built by the MPC. In fact, a partial case of the MPC is a trash-and-prepare channel $\Phi[\varrho] = {\rm tr}[\varrho] \ket{\psi_{\rm MPS}} \bra{\psi_{\rm MPS}}$, where the output pure matrix product state inherits the bond dimension of the MPC. If this is the case, the MPC tensors take the form $X^{[m]}_{i_m j_m l_m l_{m+1} i'_m j'_m l'_m l'_{m+1}} = \delta_{i_m i'_m} B^{[m]j_m}_{l_m l_{m+1}} (B^{[m]j'_m}_{l'_m l'_{m+1}})^{\ast}$, where the tensors $B^{[m]j_m}_{l_m l_{m+1}}$ define the output matrix product state. Therefore, the MPC alone is expressive to cover a purely classical optimization too; however, it is more effective to use the MPC on top of the hardware entanglement booster.

\section{Trace preservation constraint} \label{section-tp}

It is the necessity to preserve the trace that makes the variational optimization of MPC challenging. Despite the fact that the MPC tensor network defines a completely positive map by construction, if the trace is not preserved, then we can get either a zero estimate for the ground energy (provided the actual ground energy is strictly positive) or an unbounded negative estimate (provided the actual ground energy is negative but finite). It is straightforward to formulate some necessary (but not sufficient) or sufficient (but not necessary) floor-specific conditions for the MPC tensor network to be a trace preserving map; however, the known necessary and sufficient condition for trace preservation is floor-nonlocal~\cite{srinivasan-2021} which prevents the possibility to modify a single-floor tensor $X^{[m]}$ without violating the trace preservation unless exponentially many requirements are met (equivalent to the matrix equation $\Phi^{\dag}[I^{\otimes N}] = I^{\otimes N}$). Here we propose an alternative approach leading to a floor-local criterion for the trace preservation.

\begin{proposition} \label{prop-identity}
A Hermitian operator $F$ acting on the Hilbert space of $N$ systems each of dimension $d$ satisfies the inequality ${\rm tr}[F^2] - \frac{1}{d^N}({\rm tr}[F] )^2 \geq 0$ and the equality holds if and only if $F = f I^{\otimes N}$ for some $f \in \mathbb{R}$.
\end{proposition}
\begin{proof}
If ${\rm tr}[F] = 0$, then the equality is apparently true and the equality takes place if and only if $F$ is a null operator, i.e., $f = 0$. Suppose ${\rm tr}[F] \neq 0$ and $F$ is positive semidefinite, then $\widetilde{F} = ({\rm tr}[F])^{-1} F$ is a density operator. The purity parameter ${\rm tr}[\widetilde{F}^2]$ of a density operator satisfies the inequality ${\rm tr}[\widetilde{F}^2] \geq \frac{1}{d^N}$ and the equality takes place if and only if $\widetilde{F}$ is a maximally mixed state, i.e., $\widetilde{F} = \frac{1}{d^N} I^{\otimes N}$ justifying the proposition statement in this case. If ${\rm tr}[F] \neq 0$ and $F$ is not positive semidefinite, then the positive semidefinite operator $|F| = \sqrt{F^2}$ satisfies the inequality ${\rm tr}[|F|^2] - \frac{1}{d^N}({\rm tr}[|F|] )^2 \geq 0$. Noticing that ${\rm tr}[|F|^2] = {\rm tr}[F^2]$ and ${\rm tr}[|F|] \geq |{\rm tr}[F]|$, we get ${\rm tr}[F^2] - \frac{1}{d^N}({\rm tr}[F] )^2 \geq 0$. The equality takes place if and only if both ${\rm tr}[|F|] = |{\rm tr}[F]|$ and $|F| = f I^{\otimes N}$ for some $f \geq 0$, which implies that all eigenvalues of $F$ have the same signature and $F = \pm f I^{\otimes N}$.
\end{proof}

Consider an MPC $\Phi$ defined by a collection of positive semidefinite operators $\{X^{[m]}\}_{m=0}^{N-1}$ also known as \emph{cores} \cite{srinivasan-2021}. To deal with the boundary tensors $X^{[0]}$ and $X^{[N-1]}$ in the same manner as with the internal ones ($m = 1, \ldots, N-2$) and make the notation uniform for all the cores of MPC, we introduce dummy tensors $T_1^{[0]} = r^{-1/2} \delta_{l_0 l'_0}$ and $B_1^{[N]} = r^{-1/2} \delta_{l_N l'_N}$ that are contracted with new 8-rank tensors $X^{[0]}$ and $X^{[N-1]}$, respectively, see Fig.~\ref{figure-variance}(a). In this section, our goal is to express the trace preservation condition in the local form (with respect to a specific floor $m$). To do that we focus on the operator $F \equiv \Phi^{\dag}[I]$ and apply Proposition~\ref{prop-identity}. 

The 0-rank tensor ${\rm tr}[F]$ is just a contraction of $X^{[m]}$ with its tensor environment that contains the contracted floors above $T_1^{[m]}$, the contracted floors below $B_1^{[m+1]}$, and an identity tensor $\delta_{i_m i'_m} \delta_{j_m j'_m}$, see Fig.~\ref{figure-variance}(b). Algebraically,
\begin{equation} \label{first-moment}
    {\rm tr}[F] = {\rm tr}\left[ X^{[m]} \left( I \otimes I \otimes T_1^{[m]} \otimes B_1^{[m+1]} \right) \right].
\end{equation}
We refer to $T_1^{[m]}$ and $B_1^{[m+1]}$ as a \emph{top} and \emph{bottom} contraction of type 1, respectively, because they define the first moment ${\rm tr}[F]$. Similarly, we introduce the top and bottom contractions $T_2^{[m]}$ and $B_2^{[m+1]}$ of type 2 that define the second moment [see Fig.~\ref{figure-variance}(c)], namely,
\begin{equation} \label{second-moment}
    {\rm tr}[F^2] = {\rm tr}\left[ {\rm tr}_1 \big[ X^{[m]\dag} \otimes X^{[m]} \big] \left( I \otimes  I \otimes  T_2^{[m]} \otimes B_2^{[m+1]} \right) \right].
\end{equation}
If $\Phi$ is trace preserving, then the top and bottom contractions of both types satisfy the normalization condition
\begin{equation}
    {\rm tr}\left[ T_1^{[m]} \otimes B_1^{[m]} \right] = {\rm tr}\left[ T_2^{[m]} \otimes B_2^{[m]} \right] = d^{N}
\end{equation}
for each $m$. It is worth mentioning that the contractions $\{T_1^{[m]}\}_{m=0}^{N}$ and $\{T_2^{[m]}\}_{m=0}^{N}$ are all efficiently calculated by propagating from top to bottom along the tensor diagram and recursively appending one floor a time. Similarly, the contractions $\{B_1^{[m]}\}_{m=0}^{N}$ and $\{B_2^{[m]}\}_{m=0}^{N}$ are efficiently calculated by propagating from bottom to top along the tensor diagram.

\begin{figure}
    \centering
    \includegraphics[width = 8.5cm]{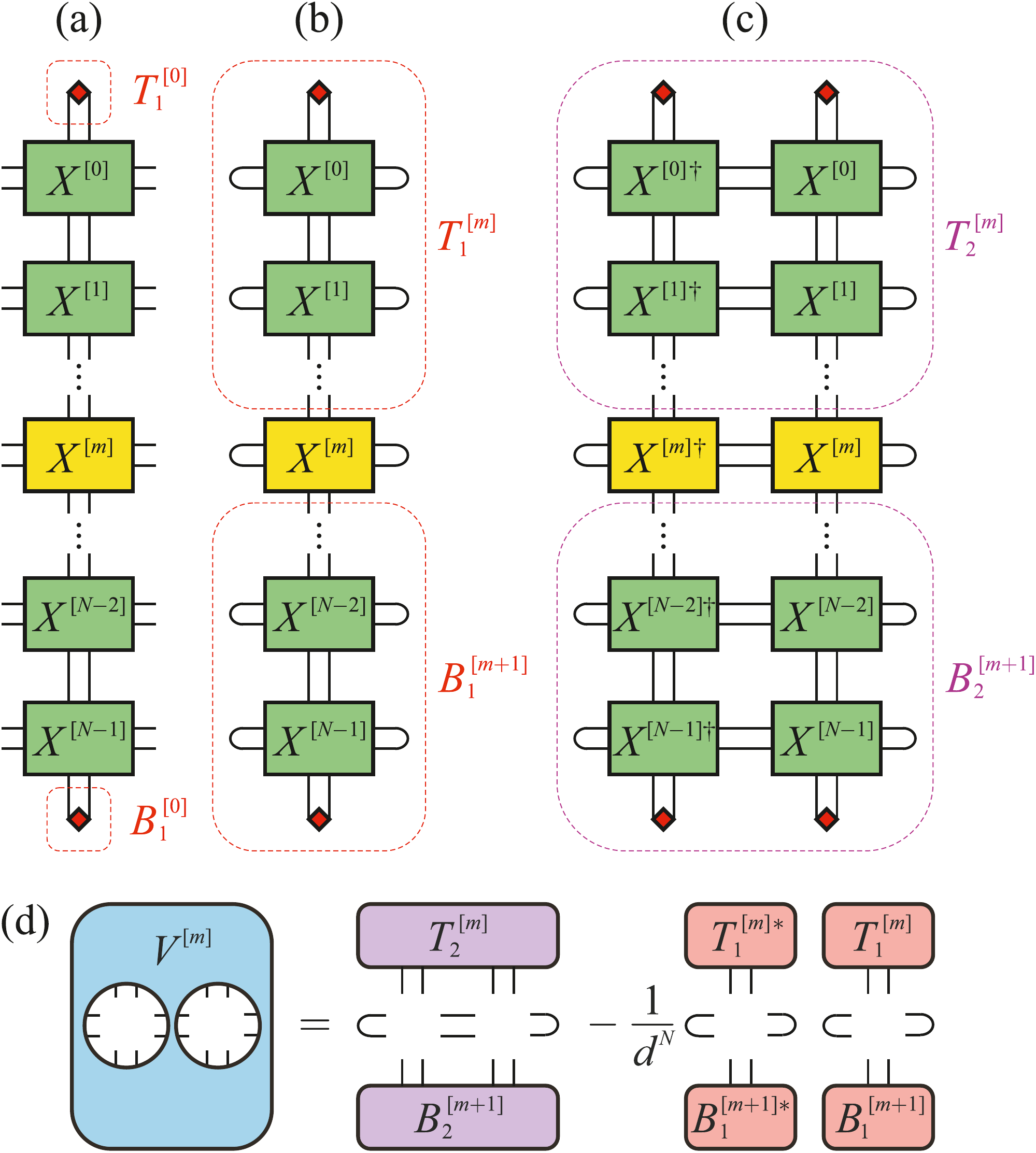}
    \caption{(a) Toward a floor-local optimization of MPC with open boundary conditions. (b) Top and bottom contractions of type 1. (c) Top and bottom contractions of type 2. (d) Variance-like tensor.}
    \label{figure-variance}
\end{figure}

Let $|X^{[m]}\rangle\rangle$ be the row-vectorized form of $X^{[m]}$, then by Proposition~\ref{prop-identity} the functional 
\begin{equation}
    {\rm tr}[F^2] - \frac{1}{d^N}({\rm tr}[F] )^2 \equiv \langle\langle X^{[m]} | V^{[m]} |X^{[m]}\rangle\rangle \geq 0
\end{equation}
is nonnegative and vanishes if and only if, on the one side, $F \equiv \Phi^{\dag}[I] \propto I^{\otimes N}$ and, on the other side, $V^{[m]} |X^{[m]}\rangle\rangle = 0$. Therefore, the introduced $d^4 r^4 \times d^4 r^4 $ matrix $V^{[m]}$ plays role of the variance for a deviation from trace preservation. Thanks to the relations \eqref{first-moment} and \eqref{second-moment}, the variance-like matrix $V^{[m]}$ is readily expressed in terms of tensor diagrams, see Fig.~\ref{figure-variance}(d). 

Adding the normalization condition, we finally obtain a floor-local criterion for an MPC to be trace preserving: 
\begin{equation} \label{criterion-TP}
    \left\{ \begin{array}{l}
    V^{[m]} |X^{[m]}\rangle\rangle = 0,  \\
    {\rm tr}\left[ X^{[m]} \left( I \otimes I \otimes T_1^{[m]} \otimes B_1^{[m+1]} \right) \right] = d^{N}.    
    \end{array} \right.
\end{equation}
Note that the conditions \eqref{criterion-TP} are linear with respect to $X^{[m]}$, which enables us to use them as additional constraints in the semidefinite programming problem to minimize the energy functional by locally modifying $X^{[m]} \geq 0$. The details on the sweeping optimization are presented in the next section.

\section{Variational optimization of a matrix product channel} \label{section-variational-algorithm}

The proposed algorithm to mitigate noise and reduce errors of the variational eigensolver works as follows. Given a sufficiently large number of measurement shots, infer the training quasistate $\varrho_{\rm tr}$ at the output of the variational eigensolver in the form of a $2N$-rank tensor. As overfitting is a standard problem in machine-learning measurement-processing algorithms (see, e.g., \cite{luchnikov-2020}), one should also have an additional collection of independent measurement outcomes and infer the validation quasistate $\varrho_{\rm val}$. Represent the Hamiltonian as a $2N$-rank tensor $H$. Contract $\varrho_{\rm val}$ and $H$ to get an initial validation estimate for the ground state energy $E_{\rm est}$. Then proceed to the variational optimization of the MPC (sandwiched in between the quasistate and $H$):
\begin{enumerate}
    \item Choose the MPC bond dimension $r$, which is an optimization hyper-parameter.
    \item Initialize the MPC to be a trace preserving perturbation of the identity map. This introduces an initial jolt both for the state and the average energy but creates additional correlations among the subsystems (qubits).
    \item Recurrently calculate the top and bottom contractions of both types ($T_{\alpha}^{[m]}$ and $B_{\alpha}^{[m]}$, $\alpha = 1,2$).
    \item Start a down sweep along the MPC (the floor label $m$ increases from $0$ to $N-2$) by performing the following manipulations at the $m$-th floor:
    \begin{enumerate}
        \item Calculate the contraction $E^{[m]}$ of tensors $\varrho_{\rm tr}$, floors $\{X^{[m']}\}_{m'=0}^{m-1}$, floors $\{X^{[m']}\}_{m'=m+1}^{N-1}$, $T_1^{[0]}$, $B_1^{[N]}$ and $H$ (i.e., all the tensors contributing to the energy except for $X^{[m]}$).
        \item Calculate the variance-like matrix $V^{[m]}$.
        \item Solve the semidefinite programming problem of minimizing ${\rm tr}[E^{[m]} X^{[m]}]$ subject to the positivity constraint $X^{[m]} \geq 0$ and the linear constraints \eqref{criterion-TP}. 
        \item Update $X^{[m]}$ in accordance with the solution obtained and recalculate the two top contractions $T_{1}^{[m+1]}$ and $T_{2}^{[m+1]}$.
        \item Calculate the average energy for the validation quasistate $\varrho_{\rm val}$.
    \end{enumerate}
    \item Start an up sweep along the MPC (the floor label $m$ decreases from $N-1$ to $1$) by performing the following manipulations at the $m$-th floor:
    \begin{enumerate}
        \item Calculate the contraction $E^{[m]}$ of tensors $R$, floors $\{X^{[m']}\}_{m'=0}^{m-1}$, floors $\{X^{[m']}\}_{m'=m+1}^{N-1}$, $T_1^{[0]}$, $B_1^{[N]}$ and $H$ (i.e., all the tensors contributing to the energy except for $X^{[m]}$).
        \item Calculate the variance-like matrix $V^{[m]}$.
        \item Solve the semidefinite programming problem of minimizing ${\rm tr}[E^{[m]} X^{[m]}]$ subject to the positivity constraint $X^{[m]} \geq 0$ and the linear constraints \eqref{criterion-TP}.
        \item Update $X^{[m]}$ in accordance with the solution obtained and recalculate two bottom contractions $B_{1}^{[m]}$ and $B_{2}^{[m]}$.
        \item Calculate the average energy for the validation quasistate $\varrho_{\rm val}$.
    \end{enumerate}
    \item Repeat the sweeps until the energies for the training quasistate and the validation quasistate start diverging (or as long as the computational resources allow). Accept the validation solution if the corresponding energy is less than the initial validation estimation $E_{\rm est}$.
\end{enumerate}

\section{Results and discussion} \label{section-results}

In this paper, we consider a proof-of-principle example of estimating the ground state energy of the stretched ${\rm H}_2{\rm O}$ molecule. The inter-nuclear distance is $1.7$ times larger the equillibrium one so that the exact ground state $\ket{\psi_0}$ has a tangible entanglement entropy and, consequently, a significant correlation energy~\cite{boguslawski-2012,boguslawski-2013,molina-espiritu-2015}. We consider an effective Hamiltonian $H$ in the cc-pVDZ basis with CAS(8e,6o), so we deal with $N = 12$ qubits. The Hamiltonian is a sum of $551$ Pauli operator strings (which is much less than $4^N \approx 1.7 \times 10^7$ of all Pauli operator strings).  The exact diagonalization reveals the ground state $\ket{\psi_0}$ and the ground energy $E_0$ that we use as a reference point ($E_0 = 0$). Calculating the von Neumann entropy of individual qubits in the ground state, we observe that some of the qubits are almost disentangled from others ($S < 0.02$ bits for 4 qubits), whereas other qubits are tangibly entangled ($S > 0.4$ bits for 8 qubits). The conventional DMRG algorithm with the bond dimension $1$ yields the correlation energy $E_{\rm DMRG (b.d.=1)} \approx \SI{0.178}{\hartree}$. 

Defining the total entanglement entropy as the maximum von Neumann entropy of a linear subsystem,
\begin{equation}
S_{\rm ent}(\psi_0) = \max_k S_{\text{~}k\text{~top~qubits~}|\text{~}N-k\text{~bottom~qubits~}}(\psi_0),
\end{equation}

\noindent we see that $S_{\rm ent}(\psi_0) \approx 1.02$~bits. In terms of the matrix product representation~\cite{schollwock-2011,orus-2014,cirac-2021}, this corresponds to the effective bond dimension $2^{S_{\rm ent}(\psi_0)} \approx 2.03$ so that any approximation of the ground state with a matrix product state of bond dimension $d_{\rm MPS} < 2^{S_{\rm ent}(\psi_0)}$ will inevitably lead to a significant error. Indeed, in the case $d_{\rm MPS} = 2$, the conventional DMRG algorithm with the bond dimension $2$ yields the error $\epsilon_{\rm cl} \equiv E_{\rm DMRG (b.d.=2)} \approx 0.096$Ha.

The VQE ansatz in Fig.~\ref{figure-circuit} is hardware efficient but rather imprecise. Although it formally creates a matrix product state with the bond dimension $8$, the variationally optimized energy is as high as $\epsilon_{\rm VQE} \approx 0.139$Ha. The effective bond dimension of this VQE state $2^{S_{\rm ent}(\psi_{\rm VQE})} \approx 1.34$, so there is definitely a room for improvement (to be made by the MPC). 

To illustrate the effect of noise on the VQE performance let us consider a noisy version of the quantum circuit in Fig.~\ref{figure-circuit} with both coherent and incoherent noise. The coherent noise still plays a major role in up-to-date quantum processors~\cite{cenedese-2022} and we simulate it by adding a normally distributed random variable ${\cal N}(0, 10^{-2} \text{~radians})$ to each of the angles in the single-qubit rotation gates. In our realization this leads to the energy increase $\SI{2}{\milli\hartree}$. To simulate the incoherent noise, each two-qubit CNOT gate is accompanied by a stochastic map described by a one-parameter Pauli-Lindblad model $e^{\cal L}$, ${\cal L}(\varrho) = \lambda \sum_{k = \{0,1,2,3\}^2} (P_k \varrho P_k - \varrho)$~\cite{berg-2022}. Here $P_k$ is a Pauli string operator (acting on the Hilbert space of two qubits). In our \emph{in silico} simulation we assume $\lambda$ to be uniformly distributed so that its average and range are both $10^{-5}$. This results in an extra energy increase of $\SI{3}{\milli\hartree}$. Equipped with such a noise model, we therefore get the noise-induced error in the energy estimation $\epsilon_{\rm noise} \approx \SI{5}{\milli\hartree}$. 

The execution of the MPC optimization algorithm (Section~\ref{section-variational-algorithm}) induces the initial jolt of the noisy state and its energy (additional $\SI{2}{\milli\hartree}$ in our case). Fig.~\ref{figure-results} depicts the total energy increment $\SI{7}{\milli\hartree}$ for the noiseless VQE energy as a result of the noise and the initial jolt. We use the training and validation sets of measurement outcomes each containing $10^8$ measurement shots. This corresponds to the standard deviation $\SI{1}{\milli\hartree}$ in the energy estimation. Running the variational optimization with the bond dimension $r=2$ we see the decrease of the average energy both for the training and the validation sets of measurement outcomes. The training and validation curves go below the DMRG value for the bond dimension $2$ ($\epsilon_{\rm cl} = \SI{0.096}{\hartree}$), thus showing an advantage of the use of the classical postprocessing algorithm on top of the quantum hardware state preparator. The divergence of the training curve from the validation one happens at the energy $\epsilon_{\rm q+cl} = \SI{0.075}{\hartree}$ which is the final estimation for this VQE ansatz and the MPC of bond dimension $r=2$. The training curve further goes below the ground state which indicates the presence of negative eigenvalues in the training quasistate due to finite statistics. Yet, since the MPC is not correlated with $\varrho_{\rm val}$ ($\varrho_{\rm val}$ is not used for the optimization of the map), the resulting energy for the validation set is an unbiased estimator of the energy of the output state. Therefore, it is guaranteed to be above the ground state energy in expectation.

\begin{figure}
    \centering
    \includegraphics[width = 9cm]{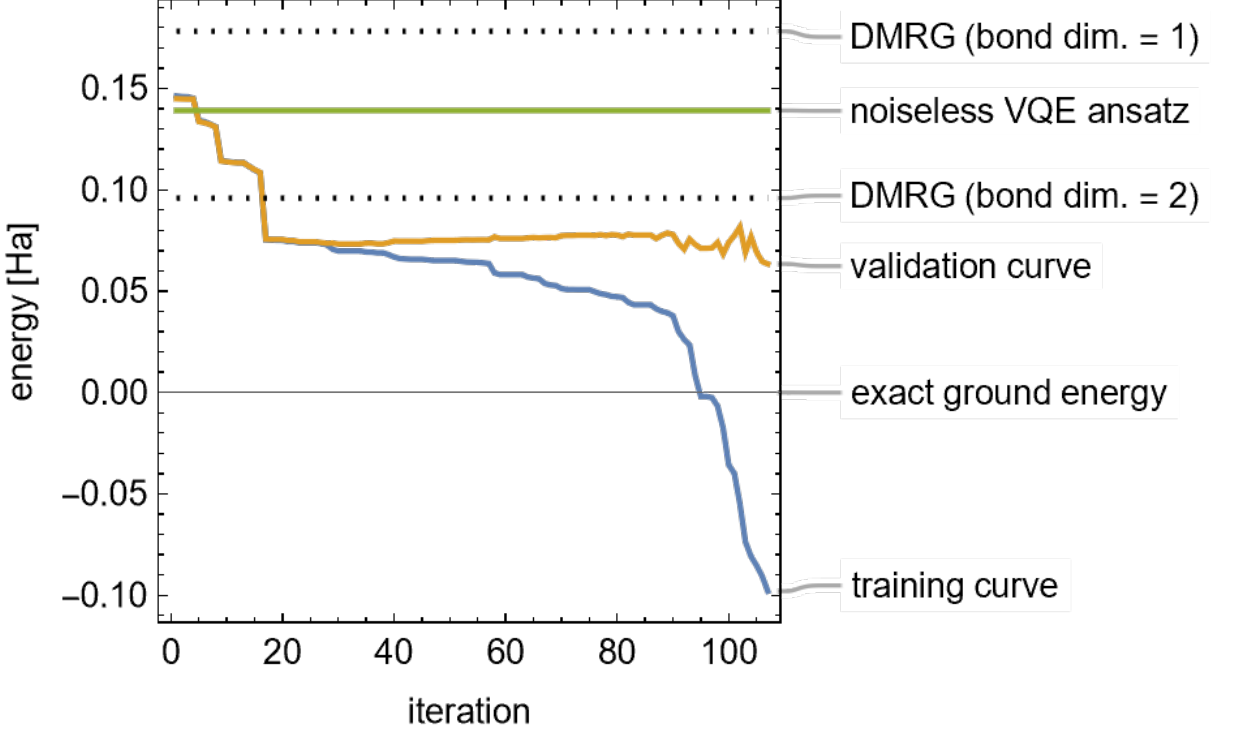}
    \caption{Performance of the MPC variational optimization algorithm.}
    \label{figure-results}
\end{figure}

Clearly, the VQE ansatz used in this example can be changed and further improved at the expense of the increased circuit depth and the increased noise. Similarly, one could use the MPC with a higher bond dimension at the expense of the more complicated semidefinite programming problem to solve. However, the greater the expressivity of the VQE ansatz and the MPC, the less trainable they are expected to be~\cite{holmes-2022}. This eventually restricts the use of too high bond dimensions in the MPC. Nevertheless, the MPC with a modest bond dimension can still be helpful in denoising the least correlated qubits (by mapping them to a pure states) and creating new bonds among the qubits. The presented variational optimization is performed in the gradient-free manner, but we do not exclude the possibility to modify the gradient-based MPC optimization~\cite{torlai-2020,guo-2022} to include the derived trace preservation constraint and thus implement another variational optimization algorithm. Further steps to alleviate the normalization condition in the trace preservation property may include forcing the MPC to take the mixed-canonical form in full analogy with that form for the matrix product states~\cite{schollwock-2011,orus-2014,cirac-2021}.

Another limitation of an MPC is attributed to the fact that it represents a completely positive map and, therefore, it cannot cancel an arbitrary noise~\cite{holevo-2019}. An extension to the realm of positive maps would be helpful, but there is no explicit characterization of positive maps in high dimensions. Nonetheless, if the noise is reasonably small, then the MPC can create new and enforce existing correlations among the qubits, thus diminishing the energy toward the ground one. Therefore, the MPC can be used on top of other noise-mitigation strategies assuming the known noise model, e.g., the probabilistic error cancellation~\cite{temme-2017,li-2017}.

\section{Conclusions} \label{section-conclusions}

We have considered a modification of the virtual linear map algorithm~\cite{vilma} where the role of classical postprocessing is played by the MPC tensor network. We have clarified the MPC expressivity and derived an alternative trace preserving condition enabling a sweeping optimization of the network. The MPC with a modest bond dimension is, on the one hand, sufficiently expressive to appropriately restore the effective bond dimension of the noisy VQE output and, on the other hand, is not expressive enough to render the variational optimization intractable. 

The complexity of MPC variational optimization scales linearly with the number of qubits used (provided the number of measurement shots and Hamiltonian components remains the same). This favorable scaling is of particular importance for near-term quantum computing when we want to scale to problem sizes relevant in modelling protein–ligand interaction energies for drug design. 

By studying the example of the stretched ${\rm H}_2{\rm O}$ molecule with a tangible entanglement we show how the MPC optimization can not only mitigate the noise of the VQE ansatz but also to amend the ansatz itself by creating additional bonds among the qubits. The proposed combination of the quantum hardware with the classical postprocessing (of a fixed bond dimension) diminishes the error down to $\epsilon_{\rm q+cl} < \epsilon_{\rm cl}$, where $\epsilon_{\rm cl}$ is the error obtained in a purely classical simulation with the same bond dimension.

\begin{acknowledgements}
The authors thank Stefan Knecht, Anton Nyk\"{a}nen, and Viacheslav Dubovitskii for help in programming the code and helpful references. 
\end{acknowledgements}

\bibliography{bibliography}

\end{document}